\documentclass[final,notitlepage,10pt,reqno,graphicx]{amsart}
\usepackage{setspace}
\usepackage{graphicx}
\usepackage{CJK}
\usepackage{rotating}
\usepackage{xcolor}
\usepackage[pstarrows]{pict2e}

\usepackage{mathptmx}       
\DeclareSymbolFont{operators}{OT1}{txr}{m}{n}
\SetSymbolFont{operators}{bold}{OT1}{txr}{bx}{n}
\def\operator@font{\mathgroup\symoperators}
\DeclareSymbolFont{italic}{OT1}{txr}{m}{it}
\SetSymbolFont{italic}{bold}{OT1}{txr}{bx}{it}
\DeclareSymbolFontAlphabet{\mathrm}{operators}
\DeclareMathAlphabet{\mathbf}{OT1}{txr}{bx}{n}
\DeclareMathAlphabet{\mathit}{OT1}{txr}{m}{it}
\SetMathAlphabet{\mathit}{bold}{OT1}{txr}{bx}{it}
\DeclareSymbolFont{letters}{OML}{txmi}{m}{it}
\SetSymbolFont{letters}{bold}{OML}{txmi}{bx}{it}
\DeclareFontSubstitution{OML}{txmi}{m}{it}
\DeclareSymbolFont{lettersA}{U}{txmia}{m}{it}
\SetSymbolFont{lettersA}{bold}{U}{txmia}{bx}{it}
\DeclareFontSubstitution{U}{txmia}{m}{it}
\DeclareSymbolFontAlphabet{\mathfrak}{lettersA}
\DeclareSymbolFont{symbols}{OMS}{txsy}{m}{n}
\SetSymbolFont{symbols}{bold}{OMS}{txsy}{bx}{n}
\DeclareFontSubstitution{OMS}{txsy}{m}{n}

\usepackage{microtype}
\usepackage[normalem]{ulem}
\usepackage{mathrsfs}
\usepackage{enumitem}[2011/09/28]

\usepackage[square]{natbib}
\bibpunct{[}{]}{,}{n}{}{;}

\usepackage{url}
\usepackage{caption}[2012/02/19]

\usepackage{tikz}
\tikzset{>=stealth}
\usepackage{pgfplots}
\pgfplotsset{scaled y ticks=false}
\usepackage{amssymb,bm,amsmath}
\usepackage{amsfonts}

\renewenvironment{abstract}
 {\small
  \begin{center}
  \bfseries \abstractname\vspace{-.5em}\vspace{0pt}
  \end{center}
  \list{}{
    \setlength{\leftmargin}{0cm}%
    \setlength{\rightmargin}{\leftmargin}%
  }%
  \item\relax}
 {\endlist}

\DeclareMathOperator{\D}{d}

\theoremstyle{plain}
\newtheorem{theorem}{Theorem}[section]

\newtheorem{lemma}[theorem]{Lemma}

\newenvironment{remark}[1][Remark]{\begin{trivlist}
\item[\hskip \labelsep {\bfseries #1}]}{\end{trivlist}}

\begin{document}

\title[Symmetries of double integrals]{Symmetries of certain double integrals\\ related to Hall effect devices}
\author{Udo Ausserlechner\and M. Lawrence Glasser \and Yajun Zhou}
\address{Infineon Technologies Austria AG, Siemensstrasse 2, Villach 9500, Austria}
\email{udo.ausserlechner@infineon.com}

\address{Dpto.~de F\'isica Te\'orica, Facultad de Ciencias, Universidad de Valladolid, Paseo Bel\'en 9, 47011 Valladolid, Spain;
Department of Physics, Clarkson University, Potsdam, NY  13699, USA}

\email{laryg@clarkson.edu}

\address{Program in Applied and Computational Mathematics (PACM), Princeton University, Princeton, NJ 08544, USA; Academy of Advanced Interdisciplinary Sciences (AAIS), Peking University, Beijing 100871, P. R. China }
\email{yajunz@math.princeton.edu, yajun.zhou.1982@pku.edu.cn}
\date{\today}
\thanks{\indent\textit{Subject Classification (AMS 2010)}: 33E05\ (Primary), 78A35 (Secondary)}
\maketitle

\begin{abstract}
    One encounters iterated elliptic integrals in the study of Hall effect devices, as a result of conformal mappings of Schwarz--Christoffel type. Some of these double elliptic integrals possess amazing symmetries with regard to the physical parameters of the underlying Hall effect devices. We give a unified mathematical treatment of such symmetric double integrals, in the context of Hall effect devices with three and four contacts.\end{abstract}

\section{Introduction}
As one can easily demonstrate  oneself, if you spin a coin, oriented perpendicular to an inclined plane,  due to the balance of gravity and the gyroscopic force, the coin will move across the plane rather than down it as it does when it is not spinning. The speed at which it moves is determined by various factors such as the tilt of the plane, the rate of spin and the surface conditions. The electrical analogue is the Hall effect: if an electron current is produced, by electrical contacts, across a conducting plate in a perpendicular magnetic field a current $I_H$, and equivalently, a voltage $V_H$,  resulting from the balance between the strength of  the current and the Lorentz force on the electrons, will be detectable between electrodes placed perpendicular to the current. The magnitude of this voltage will depend on the magnetic field strength,  the electrical characteristics of the plate material and its geometry.

 For such a standard four-contact commercial semiconductor Hall device, having two perpendicular reflection lines, one of us \cite{Ausserlechner2015,Ausserlechner2016a,Ausserlechner2017} determined the analytic form of its geometrical factor $G_H$, in the expression for $V_H$, in terms of a double elliptic integral whose two moduli depended on adjustable characteristics of the system. From numerical evaluations of conformal transformations, it was found that $G_H$ exhibited an invariance which could be expressed
as\begin{align}A(p,q):={}&\int_0^{\pi}\frac{\D x}{\sqrt{1-p\cos x}}\int_0^x\frac{\D y}{\sqrt{1+q\cos y}}
\notag\\={}&A(p',q'),\quad\forall p,q\in[0,1], \label{eq:Apq}\end{align}
where  $ p'=\sqrt{1-p^2},q'=\sqrt{1-q^2}$ are  complementary moduli.

While easily verified numerically, a proof of \eqref{eq:Apq} was, after some delay, finally presented by two of us  \cite{GlasserZhou2017} on the basis of somewhat recondite integral manipulations.        Shortly afterwards                                                                                                              David Broadhurst and Wadim Zudilin gave a different proof \cite{BroadhurstZudilin2017} for the diagonal case $A(p,p) =A(\sqrt{1-p^2},\sqrt{1-p^2})$, and discussed its arithmetic implications.

Since then a similar investigation of the three-contact Hall devices, but still possessing mirror symmetry, to be described in \S\ref{sec:phys}, has been carried out.           The study of   these novel Hall devices based on Schwarz--Christoffel conformal mappings has led one of us \cite{Ausserlechner2018} to a seemingly more complicated elliptic identity
\begin{align}&
I(\alpha,\beta)\notag\\:={}&\int_0^{\pi/2}\frac{\D \theta}{\sqrt{1-\smash[b]{(1-\alpha)\sin^2\theta}}}\int_0^\theta\frac{\D\phi}{\sqrt{1-\smash[b]{(1-\beta)\sin^2\phi}}}\notag\\{}&-\int_0^{\pi/2}\frac{\sqrt{\alpha(1-\alpha)}\sin\theta\D \theta}{\sqrt{\alpha \smash[b]{ (1-\beta )-(1-\alpha ) \beta  \cos ^2\theta}}\sqrt{1-\smash[b]{(1-\alpha)\sin^2\theta}} }\int_0^\theta\frac{\D\phi}{\sqrt{1-\smash[b]{(1-\alpha)\sin^2\phi}}}\notag\\={}&I(1-\beta,1-\alpha),\quad 0\leq\beta\leq \alpha\leq 1.\label{eq:Iab}
\end{align}

The aim of this article is to offer a mathematical proof of \eqref{eq:Iab} which will be presented in \S\ref{sec:int_id}.
Our major analytic tool in this article is a modest extension of the inhomogeneous differential equation studied by Broadhurst and  Zudilin \cite{BroadhurstZudilin2017}, in the context of $ A(p,p)$. This approach not only allows us to simplify the original proof of  \eqref{eq:Apq} published in   \cite{GlasserZhou2017}, but also sets  \eqref{eq:Apq}  and \eqref{eq:Iab} in a unified framework.

    \section{Physical Background\label{sec:phys}}The classic Hall plates detect a magnetic field orthogonal to the surface of a semiconductor. They have four contacts, whereby current is forced through two non-neighboring contacts and the output voltage is tapped at the other two contacts. At zero magnetic field the electrical behavior of the device is given by an equivalent resistor circuit (ERC) with four terminals. At small magnetic field the change of output voltage is proportional to the magnetic field, the input current, the Hall mobility, the sheet resistance, and a Hall geometry factor  \begin{align} G_{H0}^{(4C)}=\frac{1}{{\mathbf K}'\left(\frac{1-p}{1+p}\right){\mathbf K}\left(\frac{1-f}{1+f}\right)}\int_0^1\frac{\int ^x_0\frac{\D y}{\sqrt{1-y^2}\sqrt{1-\left(\frac{1-p}{1+p}\right)^2y^2}}}{\sqrt{1-x^2}\sqrt{1-\left[1-\left(\frac{1-f}{1+f}\right)^2\right](1-x^{2})}}\D x.\end{align} Here, the parameters  $p$ and $f$ are determined by the input and output resistances, and \begin{align}
\mathbf K(\sqrt{\lambda})\equiv\mathbf K'(\sqrt{1-\lambda}):=\int_0^{\pi/2}\frac{\D\theta}{\sqrt{1-\smash[b]{\lambda\sin^2\theta}}},\quad \lambda\in[0,1)
\end{align}is the complete elliptic integral of the first kind. The Hall geometry factor accounts for the shape of the Hall plate (\textit{i.e.}~its layout) and the size of the contacts.   The quantity   $ G_{H0}^{(4C)}$ can be computed as a function of geometrical parameters of the Hall plate, but it can also be expressed as a function of the resistances in the ERC \cite{Ausserlechner2018}. The thermal noise of a Hall plate at small magnetic field is also described by the ERC. Thus, the signal-to-noise ratio (SNR) of Hall plates relates in a very general way to the ERC. A numerical study of \begin{align}
\mathrm{SNR}\propto\frac{G_{H0}^{(4C)}\sqrt{{{\mathbf K'}(f)\mathbf K}(p)}}{\sqrt{{\mathbf K}(f){\mathbf K'}(p)}}
\end{align} reveals a symmetry: for every Hall plate with small contacts there is another Hall plate with properly chosen large contacts having the same SNR \cite{Ausserlechner2017}. If the Hall plate has $90^{\circ}$ symmetry like a Greek cross or an octagon, numerical evidence suggested that the SNR remains the same for the complementary device, where contacts and isolating boundaries are swapped. Both statements are equivalent to   \eqref{eq:Apq}, and they can be proven rigorously \cite{GlasserZhou2017,BroadhurstZudilin2017}.

\begin{figure}[t]
\begin{minipage}{0.5\textwidth}\fontfamily{ptm}\selectfont\begin{tikzpicture}[scale=0.75]\draw[color=gray!10!white,fill] (0,0)--(0,-1)--(3,{3*tan(15))-1})--(3,{3*tan(15)});
\draw[color=gray!30!white,fill] (0,0)--(0,-1)--(-2,{2*tan(30))-1})--(-2,{2*tan(30)});
\draw[color=gray!50!white,fill] (0,0)--(-2,{2*tan(30)})--(1,{2*tan(30)+3*tan(15)})--(3,{3*tan(15)});
\draw[color=magenta,fill] (.5,{.5*tan(15)})--({.5-1},{.5*tan(15)+1*tan(30)})--(1.5,{1*tan(30)+2.5*tan(15)})--(2.5,{2.5*tan(15)});
\draw[color=cyan,fill] (1.45,{1.45*tan(15)})--({1.45-.9},{1.45*tan(15)+.9*tan(30)})--({1.55-.9},{.9*tan(30)+1.55*tan(15)})--(1.55,{1.55*tan(15)});
\draw[color=cyan,fill] (.75,{.75*tan(15)})--({.75-.9},{.75*tan(15)+.9*tan(30)})--({1.2-.9},{.9*tan(30)+1.2*tan(15)})--(1.2,{1.2*tan(15)});
\draw[color=cyan,fill] (2.25,{2.25*tan(15)})--({2.25-.9},{2.25*tan(15)+.9*tan(30)})--({1.8-.9},{.9*tan(30)+1.8*tan(15)})--(1.8,{1.8*tan(15)});
\draw[color=magenta,fill] (.6,{.6*tan(15)-0.5})--(.5,{.5*tan(15)-0.4})--(.5,{.5*tan(15)})--(2.5,{2.5*tan(15)})--(2.5,{2.5*tan(15)-0.4})--(2.4,{2.4*tan(15)-0.5});

\draw[fill=black,color=black] (1,.69) circle[radius=0.04];
\draw[fill=black,color=black] (1.8,.69) circle[radius=0.04];
\draw[fill=black,color=black] (0.2,.69) circle[radius=0.04];
\draw(1,.69)--(2,2.1);\draw(1.8,.69)--(2.1,2.1);\draw(0.2,.69)--(1.9,2.1);
\node[above] at (2.1,2) {\begin{tiny}Three $n^+$-doped contacts\end{tiny}};

\draw[fill=black,color=black] (1,0) circle[radius=0.04];
\draw(1,0)--(1.3,-1.3);\node[below] at (1.3,-1.2){\begin{tiny}$ n^-$-doped Hall effect region\end{tiny}};

\foreach \y in {-1,-.5,0,0.5,1}\draw[->,color=blue,ultra thick]({\y+1.5+1},{1.5*tan(15)-tan(30)+\y*tan(15)})--({\y+1.5},{1.5*tan(15)+\y*tan(15)});
\node[color=blue] at (1.45,-.8) {\begin{tiny}\begin{rotate}{15}Detectable magnetic\end{rotate}\end{tiny}};
\node[color=blue] at (1.55,-1.1) {\begin{tiny}\begin{rotate}{15}field $ B_y$\end{rotate}\end{tiny}};

\draw[->](-1.85,-1.5)--(-1.85,-.75) node[above]{$z$};
\draw[->](-1.85,-1.5)--({-1.85+0.75*cos(15)},{-1.5+0.75*sin(15)}) node[right]{$x$};
\draw[->](-1.85,-1.5)--({-1.85-0.75*cos(30)},{-1.5+0.75*sin(30)}) node[left]{$y$};

\node at (-1.5,0.5) {\begin{tiny}\begin{rotate}{-30}$p$-doped\end{rotate}\end{tiny}};
\node at (-1.5,0.25) {\begin{tiny}\begin{rotate}{-30}silicon\end{rotate}\end{tiny}};
\node at (-1.7,0.1) {\begin{tiny}\begin{rotate}{-30}substrate\end{rotate}\end{tiny}};
\node at (-1.5,1) {\begin{tiny}\begin{rotate}{15}Accessible top surface\end{rotate}\end{tiny}};
\node at (-1.35,0.77) {\begin{tiny}\begin{rotate}{15}of substrate\end{rotate}\end{tiny}};

\end{tikzpicture}\end{minipage}
\begin{minipage}{0.2\textwidth}\fontfamily{ptm}\selectfont\begin{picture}(100,60)(0,-10)
\put(-15.5,37.5){\color{gray}{$ C_1$}}
\put(65.5,37.5){\color{gray}{$ C_3$}}
\put(65.5,-2.5){\color{gray}{$ C_2$}}
\put(-1.5,40){\circle{3}}
\put(61.5,40){\circle{3}}
\put(61.5,0){\circle{3}}
\put(0,40){\line(1,0){20}}
\put(25,45){$ R_{\mathrm e}$}
\put(20,42.5){\line(1,0){20}}
\put(20,42.5){\line(0,-1){5}}
\put(20,37.5){\line(1,0){20}}
\put(40,37.5){\line(0,1){5}}
\put(40,40){\line(1,0){20}}
\put(10,40){\line(0,-1){10}}
\put(-3,17.5){$ R_{\mathrm d}$}
\put(12.5,30){\line(0,-1){20}}
\put(12.5,30){\line(-1,0){5}}
\put(7.5,30){\line(0,-1){20}}
\put(7.5,10){\line(1,0){5}}
\put(10,10){\line(0,-1){10}}
\put(50,40){\line(0,-1){10}}
\put(53.5,17.5){$ R_{\mathrm d}$}
\put(52.5,30){\line(0,-1){20}}
\put(52.5,30){\line(-1,0){5}}
\put(47.5,30){\line(0,-1){20}}
\put(47.5,10){\line(1,0){5}}
\put(50,10){\line(0,-1){10}}
\put(10,0){\line(1,0){50}}
\end{picture}\end{minipage}

\begin{minipage}{0.5\textwidth}\fontfamily{ptm}\begin{center}\selectfont(a)\end{center}\end{minipage}
\begin{minipage}{0.2\textwidth}\fontfamily{ptm}\selectfont\begin{center}(b)\end{center}\end{minipage}
\caption{(Adapted from   \cite{Ausserlechner2016a}) \textbf{a} Vertical Hall effect device with three contacts and a single mirror symmetry.
\textbf{b} The equivalent resistor circuit of the device at zero magnetic field: $C_1$ and $C_3$ are the outer contacts, $C_2$ is the inner contact.
 \label{fig:3C}}\end{figure}
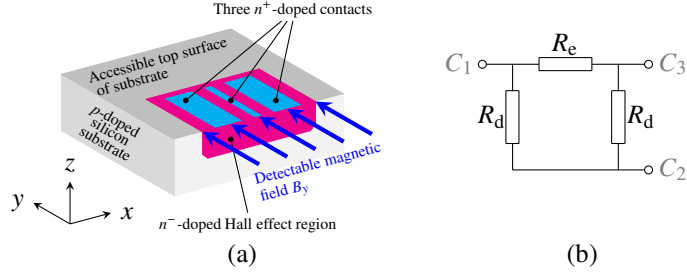

Recently, non-classical Hall devices are getting more attention, because they can detect magnetic fields parallel to the surface of the semiconductor --- they are known as Vertical Hall devices. The smallest ones have only three contacts as shown in Fig.~\ref{fig:3C}  \cite{Ausserlechner2016a}. If current is forced to flow between any two of its contacts, the output voltage at the third contact changes with magnetic field. This magnetic field sensitivity is similar to the case of traditional Hall plates, only the geometry factor $     G_{H0}^{(3C)} $ is different \cite{Ausserlechner2016a}. In contrast to $     G_{H0}^{(4C)} $,   the low field Hall geometry factor  $     G_{H0}^{(3C)} $    of devices with three contacts is a function of the resistances $R_{\mathrm e}$ and $R_{\mathrm d} $ of the ERC plus the sheet resistance $ R_{\mathrm {sh}}$.  For the case of devices having single mirror symmetry it is given explicitly in \cite{Ausserlechner2018} that \begin{align}G_{H0}^{(3C)}
=\frac{2I(\alpha,\beta)}{\mathbf K(\sqrt{\alpha})\mathbf K(\sqrt{\beta})},
\end{align} where the double integral representation for $I(\alpha,\beta) $  is given as the first equality in \eqref{eq:Iab}, and the relations \begin{align}
\frac{\mathbf K'(\sqrt{\alpha})}{\mathbf K(\sqrt{\alpha})}=\frac{R_{\mathrm e}R_{\mathrm d}}{(R_{\mathrm e}+2R_{\mathrm d})R_{\mathrm {sh}}},\quad \frac{\mathbf K'(\sqrt{\beta})}{\mathbf K(\sqrt{\beta})}=\frac{R_{\mathrm d}}{R_{\mathrm {sh}}}
\end{align}define the parameters $ \alpha,\beta$ in terms of effective resistances of the device (Fig.~\ref{fig:3C}). The corresponding SNR is proportional to $ I(\alpha,\beta)\big/\sqrt{\mathbf K(\sqrt{\alpha})\mathbf K'(\sqrt{\alpha})\mathbf K(\sqrt{\beta})\mathbf K'(\sqrt{\beta})}$. Hall devices with three contacts are conjectured to have the same symmetry property as $90^\circ$ symmetric Hall plates with four contacts:
\begin{quote}\textit{Such Hall effect devices have  the same SNR as their complementary devices.}\end{quote}

In other words, numerical experiments have suggested that $ I(\alpha,\beta)=I(1-\beta,1-\alpha)$,
        in the notation of  \eqref{eq:Iab}. The rest of the paper gives a mathematical proof of this symmetry.

        It is interesting to formulate, from a more philosophical point of
view, a rationale for the emergence of elliptic symmetries from the
intrinsic properties of these Hall devices and the possibility of predicting others. Two salient features which may be critical are the
existence of reflection symmetry in the device geometry and the
presence of the magnetic field---a pseudo-vector---which reverses under mirror reflection. It may be that any such structure will be
fruitful in this regard.

\section{Transformations of Certain Double Integrals\label{sec:int_id}}

\begin{lemma} For $\alpha,\beta\in(0,1) $, we have \begin{subequations}\allowdisplaybreaks{\begin{align}{}&\int
_0^{\pi/2}\frac{\D \theta}{\sqrt{1-\smash[b]{(1-\alpha)\sin^2\theta}}}\int_0^\theta\frac{\D\phi}{\sqrt{1-\smash[b]{(1-\beta)\sin^2\phi}}}
\notag\\={}&\frac{1}{\pi}\int_0^\beta\frac{\mathbf K(\sqrt{1-\smash[b]{\beta}})\mathbf K(\sqrt{t})}{\sqrt{t}+\sqrt{\vphantom{1}\alpha}}\frac{\D t}{\sqrt{t}}+\frac{1}{\pi}\int_\beta^1\frac{\mathbf K(\sqrt{\smash[b]{\beta}})\mathbf K(\sqrt{1-t})}{\sqrt{t}+\sqrt{\vphantom{1}\alpha}}\frac{\D t}{\sqrt{t}}\label{eq:I1a}\\={}&\mathbf K(\sqrt{1-\alpha})\mathbf K(\sqrt{1-\beta})-\frac{1}{\pi}\int_0^\alpha\frac{\mathbf K(\sqrt{1-\smash[b]{\alpha}})\mathbf K(\sqrt{t})}{\sqrt{t}+\sqrt{\beta}}\frac{\D t}{\sqrt{t}}\notag\\{}&-\frac{1}{\pi}\int_\alpha^1\frac{\mathbf K(\sqrt{\smash[b]{\alpha}})\mathbf K(\sqrt{1-t})}{\sqrt{t}+\sqrt{\vphantom{1}\beta}}\frac{\D t}{\sqrt{t}};\label{eq:I1b}\end{align}}for $ 0<\beta<\alpha<1$, we have \end{subequations}
\begin{align}&\int_0^{\pi/2}\frac{\sqrt{\alpha(1-\alpha)}\sin\theta\D \theta}{\sqrt{\alpha \smash[b]{ (1-\beta )-(1-\alpha ) \beta  \cos ^2\theta}}\sqrt{1-\smash[b]{(1-\alpha)\sin^2\theta}} }\int_0^\theta\frac{\D\phi}{\sqrt{1-\smash[b]{(1-\alpha)\sin^2\phi}}}\notag\\={}&-\mathscr P\int_0^1\frac{\mathbf K(\sqrt{\smash[b]{\beta}})\mathbf K(\sqrt{1-t})\D t}{\pi(\alpha -t)}-\frac{1}{\pi}\int_0^\beta\frac{\mathbf K(\sqrt{1-\smash[b]{\beta}})\mathbf K(\sqrt{t})-\mathbf K(\sqrt{\smash[b]{\beta}})\mathbf K(\sqrt{1-t})}{\alpha -t}\D t\notag\\{}&+\frac{1}{\pi}\int_0^\beta\frac{\mathbf K(\sqrt{1-\smash[b]{\beta}})\mathbf K(\sqrt{t})}{\sqrt{1-t}+\sqrt{\vphantom{1}1-\alpha}}\frac{\D t}{\sqrt{1-t}}+\frac{1}{\pi}\int_\beta^1\frac{\mathbf K(\sqrt{\smash[b]{\beta}})\mathbf K(\sqrt{1-t})}{\sqrt{1-t}+\sqrt{\vphantom{1}1-\alpha}}\frac{\D t}{\sqrt{1-t}},\label{eq:I2}\end{align}where $ \mathscr P$ denotes Cauchy principal value.\end{lemma}\begin{proof}We note that the following differential operator\begin{align}\widehat L_\lambda:=
\frac{\partial}{\partial \lambda}\left[ \lambda(1-\lambda)\frac{\partial}{\partial \lambda} \right]-\frac{1}{4}
\end{align}annihilates both $ \mathbf K(\sqrt{\lambda}),\lambda\in(0,1)$ and $ \mathbf K(\sqrt{1-\lambda}),\lambda\in(0,1)$. The Wro\'nskian determinant for these two linearly independent solutions to the homogeneous differential equation $ \widehat L_\lambda f(\lambda)=0$  assumes the form\begin{align}
\det\begin{pmatrix}\mathbf K(\sqrt{\lambda}) & \mathbf K(\sqrt{1-\lambda}) \\
\frac{\partial\mathbf K(\sqrt{\lambda})}{\partial \lambda} & \frac{\partial\mathbf K(\sqrt{1-\lambda})}{\partial \lambda} \\
\end{pmatrix}=-\frac{\pi}{4\lambda(1-\lambda)}.\label{eq:Wronskian}
\end{align}

It is straightforward to compute that\begin{align}{}&
\widehat L_\beta\int
_0^{\pi/2}\frac{\D \theta}{\sqrt{1-\smash[b]{(1-\alpha)\sin^2\theta}}}\int_0^\theta\frac{\D\phi}{\sqrt{1-\smash[b]{(1-\beta)\sin^2\phi}}}\notag\\={}&-\frac{1}{4}\int
_0^{\pi/2}\frac{\sin\theta\cos\theta\D \theta}{\sqrt{1-\smash[b]{(1-\alpha)\sin^2\theta}}[1-(1-\beta)\sin^2\theta]^{3/2}}\notag\\={}&-\frac{1}{4(\sqrt{\alpha}+\sqrt{\beta})\sqrt{\beta}},
\label{eq:elem}\end{align}as well as \begin{align}{}&
\widehat L_\beta\left[ \frac{1}{\pi}\int_0^\beta\frac{\mathbf K(\sqrt{1-\smash[b]{\beta}})\mathbf K(\sqrt{t})}{\sqrt{t}+\sqrt{\vphantom{1}\alpha}}\frac{\D t}{\sqrt{t}}+\frac{1}{\pi}\int_\beta^1\frac{\mathbf K(\sqrt{\smash[b]{\beta}})\mathbf K(\sqrt{1-t})}{\sqrt{t}+\sqrt{\vphantom{1}\alpha}}\frac{\D t}{\sqrt{t}} \right]\notag\\={}&-\frac{1}{4(\sqrt{\alpha}+\sqrt{\beta})\sqrt{\beta}}.\label{eq:Duhamel_spec}
\end{align}Here, it takes only elementary differentiations and integrations  to verify \eqref{eq:elem}, while one can use the Wro\'nskian determinant \eqref{eq:Wronskian} to show that  \eqref{eq:Duhamel_spec}  is  a special case of\begin{align}&\widehat L_\beta\left[ \frac{1}{\pi}\int_0^\beta\mathbf K(\sqrt{1-\smash[b]{\beta}})\mathbf K(\sqrt{t})g(\alpha,t)\D t+\frac{1}{\pi}\int_\beta^1\mathbf K(\sqrt{\smash[b]{\beta}})\mathbf K(\sqrt{1-t})g(\alpha,t)\D t \right]\notag\\={}&-\frac{g(\alpha,\beta)}{4},\label{eq:Duhamel}
\end{align} for any suitably regular bivariate function $ g(\alpha,\beta)$. Therefore, the identity \eqref{eq:I1a} must be true, up to an additive term $ f_1(\alpha)\mathbf K(\sqrt{\beta})+f_2(\alpha)\mathbf K(\sqrt{1-\beta})$.  For fixed $ \alpha\in(0,1)$, the expression\begin{align}{}&
f_1(\alpha)\mathbf K(\sqrt{\beta})+f_2(\alpha)\mathbf K(\sqrt{1-\beta})\notag\\:={}&\int
_0^{\pi/2}\frac{\D \theta}{\sqrt{1-\smash[b]{(1-\alpha)\sin^2\theta}}}\int_0^\theta\frac{\D\phi}{\sqrt{1-\smash[b]{(1-\beta)\sin^2\phi}}}\notag\\{}&-\left[ \frac{1}{\pi}\int_0^\beta\frac{\mathbf K(\sqrt{1-\smash[b]{\beta}})\mathbf K(\sqrt{t})}{\sqrt{t}+\sqrt{\vphantom{1}\alpha}}\frac{\D t}{\sqrt{t}}+\frac{1}{\pi}\int_\beta^1\frac{\mathbf K(\sqrt{\smash[b]{\beta}})\mathbf K(\sqrt{1-t})}{\sqrt{t}+\sqrt{\vphantom{1}\alpha}}\frac{\D t}{\sqrt{t}} \right]
\end{align} remains finite as $ \beta\to0^+$, so we must have $ f_2(\alpha)=0$. By subsequent asymptotic analysis in the $\beta\to1^- $ regime, we can confirm $ f_1(\alpha)=0$, thereby arriving at  \eqref{eq:I1a}  in its entirety.

To deduce   \eqref{eq:I1b} from \eqref{eq:I1a}, simply notice that \allowdisplaybreaks{\begin{align}&
\int
_0^{\pi/2}\frac{\D \theta}{\sqrt{1-\smash[b]{(1-\alpha)\sin^2\theta}}}\int_0^\theta\frac{\D\phi}{\sqrt{1-\smash[b]{(1-\beta)\sin^2\phi}}}\notag\\{}&+\int
_0^{\pi/2}\frac{\D \theta}{\sqrt{1-\smash[b]{(1-\beta)\sin^2\theta}}}\int_0^\theta\frac{\D\phi}{\sqrt{1-\smash[b]{(1-\alpha)\sin^2\phi}}}\notag\\={}&\int
_0^{\pi/2}\frac{\D \theta}{\sqrt{1-\smash[b]{(1-\alpha)\sin^2\theta}}}\int_0^{\pi/2}\frac{\D\phi}{\sqrt{1-\smash[b]{(1-\beta)\sin^2\phi}}}\notag\\={}&\mathbf K(\sqrt{1-\alpha})\mathbf K(\sqrt{1-\beta}).
\end{align}}

Differentiating under the integral sign, and integrating by parts (with respect to $ \theta$), we can verify that \begin{align}{}&
\widehat L_\beta\int_0^{\pi/2}\frac{\sqrt{\alpha(1-\alpha)}\sin\theta\left[ \int_0^\theta\frac{\D\phi}{\sqrt{1-\smash[b]{(1-\alpha)\sin^2\phi}}} \right]\D \theta}{\sqrt{\alpha \smash[b]{ (1-\beta )-(1-\alpha ) \beta  \cos ^2\theta}}\sqrt{1-\smash[b]{(1-\alpha)\sin^2\theta}}}\notag\\={}&\frac{\sqrt{1-\alpha}}{4(\alpha-\beta)\sqrt{1-\beta}}=\frac{1}{4(\alpha-\beta)}-\frac{1}{4(\sqrt{1-\alpha}+\sqrt{1-\beta})\sqrt{1-\beta}}. \end{align} According to our previous experience [cf.~\eqref{eq:Duhamel}], there must exist functions $ g_1(\alpha)$ and $ g_2(\alpha)$ such that \allowdisplaybreaks{\begin{align}&
g_1(\alpha)\mathbf K(\sqrt{\beta})+g_2(\alpha)\mathbf K(\sqrt{1-\beta})\notag\\={}&\int_0^{\pi/2}\frac{\sqrt{\alpha(1-\alpha)}\sin\theta\left[ \int_0^\theta\frac{\D\phi}{\sqrt{1-\smash[b]{(1-\alpha)\sin^2\phi}}} \right]\D \theta}{\sqrt{\alpha \smash[b]{ (1-\beta )-(1-\alpha ) \beta  \cos ^2\theta}}\sqrt{1-\smash[b]{(1-\alpha)\sin^2\theta}}}\notag\\{}&+\frac{1}{\pi}\int_0^\beta\frac{\mathbf K(\sqrt{1-\smash[b]{\beta}})\mathbf K(\sqrt{t})}{\alpha -t}\D t+\mathscr P\int_\beta^1\frac{\mathbf K(\sqrt{\smash[b]{\beta}})\mathbf K(\sqrt{1-t})}{\pi(\alpha -t)}\D t\notag\\{}&-\frac{1}{\pi}\int_0^\beta\frac{\mathbf K(\sqrt{1-\smash[b]{\beta}})\mathbf K(\sqrt{t})}{\sqrt{1-t}+\sqrt{\vphantom{1}1-\alpha}}\frac{\D t}{\sqrt{1-t}}-\frac{1}{\pi}\int_\beta^1\frac{\mathbf K(\sqrt{\smash[b]{\beta}})\mathbf K(\sqrt{1-t})}{\sqrt{1-t}+\sqrt{\vphantom{1}1-\alpha}}\frac{\D t}{\sqrt{1-t}}
\end{align}}holds for $ 0<\beta<\alpha<1$. In view of the asymptotic behavior in the regime where $ \beta\to0^+$, we must have $ g_2(\alpha)=0$. Then, we explore another extreme scenario, where $ \beta\to\alpha-0^+$, and\begin{align}&
g_1(\alpha)\mathbf K(\sqrt{\alpha})\notag\\={}&\frac{[\mathbf K(\sqrt{1-\alpha})]^2}{2}+\mathscr P\int_0^\alpha\frac{\mathbf K(\sqrt{1-\smash[b]{\alpha}})\mathbf K(\sqrt{t})}{\pi(\alpha -t)}\D t+\mathscr P\int_\alpha^1\frac{\mathbf K(\sqrt{\smash[b]{\alpha}})\mathbf K(\sqrt{1-t})}{\pi(\alpha -t)}\D t\notag\\{}&-\frac{1}{\pi}\int_0^\alpha\frac{\mathbf K(\sqrt{1-\smash[b]{\alpha}})\mathbf K(\sqrt{t})}{\sqrt{1-t}+\sqrt{\vphantom{1}1-\alpha}}\frac{\D t}{\sqrt{1-t}}-\frac{1}{\pi}\int_\alpha^1\frac{\mathbf K(\sqrt{\smash[b]{\alpha}})\mathbf K(\sqrt{1-t})}{\sqrt{1-t}+\sqrt{\vphantom{1}1-\alpha}}\frac{\D t}{\sqrt{1-t}}.\label{eq:g1}
\end{align}Here, by a reflection $ t=1-s$ and a back reference to    \eqref{eq:I1a}, we obtain\begin{align}&
\frac{1}{\pi}\int_0^\alpha\frac{\mathbf K(\sqrt{1-\smash[b]{\alpha}})\mathbf K(\sqrt{t})}{\sqrt{1-t}+\sqrt{\vphantom{1}1-\alpha}}\frac{\D t}{\sqrt{1-t}}+\frac{1}{\pi}\int_\alpha^1\frac{\mathbf K(\sqrt{\smash[b]{\alpha}})\mathbf K(\sqrt{1-t})}{\sqrt{1-t}+\sqrt{\vphantom{1}1-\alpha}}\frac{\D t}{\sqrt{1-t}}\notag\\={}&\frac{1}{\pi}\int^1_{1-\alpha}\frac{\mathbf K(\sqrt{1-\smash[b]{\alpha}})\mathbf K(\sqrt{1-s})}{\sqrt{s}+\sqrt{\vphantom{1}1-\alpha}}\frac{\D s}{\sqrt{s}}+\frac{1}{\pi}\int_0^{1-\alpha}\frac{\mathbf K(\sqrt{\smash[b]{\alpha}})\mathbf K(\sqrt{s})}{\sqrt{s}+\sqrt{\vphantom{1}1-\alpha}}\frac{\D s}{\sqrt{s}}\notag\\={}&\int
_0^{\pi/2}\frac{\D \theta}{\sqrt{1-\smash[b]{\alpha\sin^2\theta}}}\int_0^\theta\frac{\D\phi}{\sqrt{1-\smash[b]{\alpha\sin^2\phi}}}=\frac{[\mathbf K(\sqrt{\alpha})]^2}{2}.
\end{align}  Thus, we may reduce \eqref{eq:g1}
into \begin{align}
g_1(\alpha)\mathbf K(\sqrt{\alpha})={}&\frac{[\mathbf K(\sqrt{1-\alpha})]^2-[\mathbf K(\sqrt{\alpha})]^2}{2}+\mathscr P\int_0^\alpha\frac{\mathbf K(\sqrt{1-\smash[b]{\alpha}})\mathbf K(\sqrt{t})}{\pi(\alpha -t)}\D t\notag\\{}&+\mathscr P\int_\alpha^1\frac{\mathbf K(\sqrt{\smash[b]{\alpha}})\mathbf K(\sqrt{1-t})}{\pi(\alpha -t)}\D t.
\end{align}We can prove the following identity for distinct $ \alpha,\beta\in(0,1)$:\begin{align}&\mathscr P
\int_0^\alpha\frac{\mathbf K(\sqrt{1-\smash[b]{\alpha}})\mathbf K(\sqrt{t})}{\pi(\beta-t)}\D t+\mathscr P
\int_\alpha^1\frac{\mathbf K(\sqrt{\smash[b]{\alpha}})\mathbf K(\sqrt{1-t})}{\pi(\beta-t)}\D t\notag\\{}&+\mathscr P
\int_0^\beta\frac{\mathbf K(\sqrt{1-\smash[b]{\beta}})\mathbf K(\sqrt{t})}{\pi(\alpha -t)}\D t+\mathscr P
\int_\beta^1\frac{\mathbf K(\sqrt{\smash[b]{\beta}})\mathbf K(\sqrt{1-t})}{\pi(\alpha-t)}\D t\notag\\{}&-\mathbf K(\sqrt{\alpha})\mathbf K(\sqrt{\beta})+\mathbf K(\sqrt{1-\alpha})\mathbf K(\sqrt{1-\beta})=0\label{eq:recip}
\end{align} by checking that its left-hand side extends to a smooth function of $\beta\in(0,1) $ that is annihilated by $ \widehat L_\beta$ (cf.~\cite[(2.1.6)]{AGF_PartII}), and remains finite as  $ \beta(1-\beta)\to0^+$. In view of this, the expression $ g_1(\alpha)\mathbf K(\sqrt{\alpha})$ must vanish identically, as we send  $ \beta\to\alpha$ in \eqref{eq:recip}. This completes the proof of \eqref{eq:I2}. \end{proof}\begin{remark}An alternative formulation of   \eqref{eq:I1a}, namely\begin{align}&
\int_0^{\pi/2}\frac{\D\theta}{\sqrt{1-\smash[b]{\alpha}\sin^2\theta}}\int_0^\theta\frac{\D\phi}{\sqrt{1-\beta\sin^2\smash[b]{\phi}}}\notag\\={}&\frac{1}{\pi}\int_0^\beta\frac{\mathbf K(\sqrt{1-\smash[b]{\beta}})\mathbf K(\sqrt{t})}{\sqrt{1-t}+\sqrt{1-\alpha}}\frac{\D t}{\sqrt{1-t}}+\frac{1}{\pi}\int_\beta^1\frac{\mathbf K(\sqrt{\smash[b]{\beta}})\mathbf K(\sqrt{1-t})}{\sqrt{1-t}+\sqrt{1-\alpha}}\frac{\D t}{\sqrt{1-t}},
\end{align}appeared in \cite[(2)]{GlasserZhou2017}, as a precursor to the proof of the symmetric identity $A(p,q)=A(\sqrt{1-p^2},\sqrt{1-q^2}) $.  Originally, \cite[(2)]{GlasserZhou2017} was built on some addition formulae of Legendre type from \cite{AGF_PartII}, which involved  heavier computations than the procedures presented in the proof above. After reading \cite{BroadhurstZudilin2017}, one of us (Y.Z.) realized that the proof of  \cite[(2)]{GlasserZhou2017}  can be simplified by  inhomogeneous differential equations, as exploited by Broadhurst and Zudilin in their proof of $ A(p,p)=A(\sqrt{1-p^2},\sqrt{1-p^2})$. Similarly, one can verify several integral identities in  \cite{AGF_PartII} (which are triple integral analogs of  \cite[(2)]{GlasserZhou2017}) by differential equations and elementary integrations, once their forms are discovered.  \end{remark}\begin{remark}Since  we have  \cite[(51)]{Zhou2013Pnu}\begin{align}
\frac{[\mathbf K(\sqrt{1-\alpha})]^2-[\mathbf K(\sqrt{\alpha})]^2}{2}=-\mathscr P\int_0^1\frac{\mathbf K(\sqrt{1-\smash[b]{t}})\mathbf K(\sqrt{t})}{\pi(\alpha -t)}\D t,
\end{align}our proof of $g_1(\alpha)\mathbf K(\sqrt{\alpha})=0 $ amounts to the following vanishing identity\begin{align}
0={}&\frac{1}{\pi}\int_0^\alpha\frac{[\mathbf K(\sqrt{1-\smash[b]{\alpha}})-\mathbf K(\sqrt{1-\smash[b]{t}})]\mathbf K(\sqrt{t})}{\alpha -t}\D t\notag\\{}&+\frac{1}{\pi}\int_\alpha^1\frac{[\mathbf K(\sqrt{\smash[b]{\alpha}})-\mathbf K(\sqrt{t})]\mathbf K(\sqrt{1-t})}{\alpha -t}\D t.
\end{align}There are many more vanishing identities of similar shape in \cite[\S3.2]{AGF_PartII}, which are relevant to the arithmetic studies of automorphic Green's functions.\end{remark}
\begin{theorem}\label{thm:ab_recip}The double integral identity in \eqref{eq:Iab} holds.\end{theorem}
\begin{proof}By now, it is clear that \begin{align}
I(\alpha,\beta)={}&\mathbf K(\sqrt{1-\alpha})\mathbf K(\sqrt{1-\beta})\notag\\{}&+\frac{1}{\pi}\int_0^\beta\frac{\mathbf K(\sqrt{1-\smash[b]{\beta}})\mathbf K(\sqrt{t})}{\alpha -t}\D t+\mathscr P\int_\beta^1\frac{\mathbf K(\sqrt{\smash[b]{\beta}})\mathbf K(\sqrt{1-t})}{\pi(\alpha -t)}\D t\notag\\{}&-\frac{1}{\pi}\int_0^\alpha\frac{\mathbf K(\sqrt{1-\smash[b]{\alpha}})\mathbf K(\sqrt{t})}{\sqrt{t}+\sqrt{\beta}}\frac{\D t}{\sqrt{t}}-\frac{1}{\pi}\int_\alpha^1\frac{\mathbf K(\sqrt{\smash[b]{\alpha}})\mathbf K(\sqrt{1-t})}{\sqrt{t}+\sqrt{\vphantom{1}\beta}}\frac{\D t}{\sqrt{t}}\notag\\{}&-\frac{1}{\pi}\int^1_{1-\beta}\frac{\mathbf K(\sqrt{1-\smash[b]{\beta}})\mathbf K(\sqrt{1-s})}{\sqrt{s}+\sqrt{\vphantom{1}1-\alpha}}\frac{\D s}{\sqrt{s}}-\frac{1}{\pi}\int_0^{1-\beta}\frac{\mathbf K(\sqrt{\smash[b]{\beta}})\mathbf K(\sqrt{s})}{\sqrt{s}+\sqrt{\vphantom{1}1-\alpha}}\frac{\D s}{\sqrt{s}},
\end{align}so we must have \begin{align}&
I(\alpha,\beta)-I(1-\beta,1-\alpha)\notag\\={}&\mathbf K(\sqrt{1-\alpha})\mathbf K(\sqrt{1-\beta})-\mathbf K(\sqrt{\alpha})\mathbf K(\sqrt{\beta})\notag\\{}&+\frac{1}{\pi}\int_0^\beta\frac{\mathbf K(\sqrt{1-\smash[b]{\beta}})\mathbf K(\sqrt{t})}{\alpha -t}\D t+\mathscr P\int_\beta^1\frac{\mathbf K(\sqrt{\smash[b]{\beta}})\mathbf K(\sqrt{1-t})}{\pi(\alpha -t)}\D t\notag\\{}&-\frac{1}{\pi}\int_0^{1-\alpha}\frac{\mathbf K(\sqrt{\alpha})\mathbf K(\sqrt{t})}{1-\beta-t}\D t-\mathscr P\int_{1-\alpha}^1\frac{\mathbf K(\sqrt{\smash[b]{1-\alpha}})\mathbf K(\sqrt{1-t})}{\pi(1-\beta\ -t)}\D t=0,
\end{align}as a consequence of \eqref{eq:recip}. Although our proof above draws on the assumption that $ 0<\beta<\alpha<1$, its validity extends to $ 0\leq\beta\leq \alpha\leq 1$, by continuity.\end{proof}


\subsection*{Acknowledgements}
Partial financial support is acknowledged to the Spanish Junta de Castilla y Le\'on (VA057U16)
and MINECO (Project MTM2014-57129-C2-1-P).


\begin{thebibliography}{1}

\bibitem{Ausserlechner2015}
Udo Ausserlechner.
\newblock A method to compute the {H}all-geometry factor at weak magnetic field
  in closed analytical form.
\newblock {\em Electrical Engineering}, 98(3):189--206, 2015.

\bibitem{Ausserlechner2016a}
Udo Ausserlechner.
\newblock Hall effect devices with three terminals: Their magnetic sensitivity
  and offset cancellation scheme.
\newblock {\em Journal of Sensors}, 2016:1--16, 2016.

\bibitem{Ausserlechner2017}
Udo Ausserlechner.
\newblock The signal-to-noise ratio and a hidden symmetry of {H}all plates.
\newblock {\em Solid-State Electronics}, 135:14--23, 2017.

\bibitem{Ausserlechner2018}
Udo Ausserlechner.
\newblock An analytical theory of {H}all-effect devices with three contacts.
\newblock {\em Open Physics Journal}, 2018.
\newblock (to appear).

\bibitem{GlasserZhou2017}
M.~Lawrence Glasser and Yajun Zhou.
\newblock A functional identity involving elliptic integrals.
\newblock {\em Ramanujan J.}, 2017.
\newblock (to appear) \url{doi:10.1007/s11139-017-9915-4} \url{arXiv:1701.06310} [math-ph].

\bibitem{BroadhurstZudilin2017}
David Broadhurst and Wadim Zudilin.
\newblock A magnetic double integral.
\newblock {\em J. Aust. Math. Soc.}, 2018.
\newblock (to appear) \url{arXiv:1708.02381} [math.NT].

\bibitem{Zhou2013Pnu}
Yajun Zhou.
\newblock Legendre functions, spherical rotations, and multiple elliptic
  integrals.
\newblock {\em Ramanujan J.}, 34:373--428, 2014 \url{arXiv:1301.1735} [math.CA].


\bibitem{AGF_PartII}
Yajun Zhou.
\newblock {K}ontsevich--{Z}agier integrals for automorphic {G}reen's functions.
  {II}.
\newblock {\em Ramanujan J.}, 42:623--688, 2017.
\newblock [See {\em Ramanujan J.} (to appear)
  \url{doi:10.1007/s11139-017-9962-x} for erratum/addendum] \url{arXiv:1506.00318} [math.NT].


\end{thebibliography}
\end{document}